\documentclass[aps,pra, superscriptaddress, twocolumn]{revtex4}

\usepackage{bbm}
\usepackage{amssymb}
\usepackage{amsmath}
\usepackage{verbatim}
\usepackage[pdftex]{graphicx}
\usepackage{dsfont}
\usepackage{color}

\newtheorem{theorem}{Theorem}
\newtheorem{definition}{Definition}

\newtheorem{corollary}{Corollary}

\newenvironment{proof}{\vspace{1.5ex}\par\noindent\textbf{Proof}}%
   {\hspace*{\fill}$\Box$\vspace{1.5ex}\par}

\newcommand{\ad}{a^{\dagger}}

\newcommand{\tr}{\mbox{tr}}

\newcommand{\vac}{|0\rangle }

\begin{document}
\title{A time-dependent variational principle for dissipative dynamics}
\author{Christina V.\ Kraus}
\affiliation{Institute for Quantum Optics and Quantum Information of the Austrian Academy of Sciences, A-6020 Innsbruck, Austria}
\affiliation{Institute for Theoretical Physics, University of Innsbruck, A-6020 Innsbruck, Austria}

\author{Tobias J.\ Osborne}
\affiliation{Leibniz Universit\"at Hannover, Institute of Theoretical Physics, Appelstrasse 2, D-30167 Hannover, Germany}

\begin{abstract}
We extend the time-dependent variational principle to the setting of dissipative dynamics. This provides a locally optimal (in time) approximation to the dynamics of any Lindblad equation within a given variational manifold of mixed states. In contrast to the pure-state setting there is no canonical information geometry for mixed states and this leads to a family of possible trajectories --- one for each information metric. We focus on the case of the operationally motivated family of \emph{monotone riemannian metrics} and show further, that in the particular case where the variational manifold is given by the set of \emph{fermionic gaussian states} all of these possible trajectories coincide. We illustrate our results in the case of the Hubbard model subject to spin decoherence.
\end{abstract}

\maketitle
\section{Introduction}
One of the main challenges in a quantum mechanical experiment is to overcome the interaction of a system with its environment. Such interactions lead to \emph{decoherence} and often obscure coherent quantum phenomena. Recently it has been shown that this vice can be turned into a virtue: dissipative processes can be exploited as a possible resource for quantum state engineering~\cite{Wolf_DSE, Diehl_Kantian, Kraus_Kantian}. Several evolutions leading to non-trivial fixed points have now been proposed, including states with non-trivial topological properties~\cite{Diehl_Rico}. Such \emph{dissipative engineering} has opened up a completely new world for us. 

Motivated by the new possibilities offered by dissipative engineering there has been renewed interest in understanding dissipative processes in more detail. However, this task is complicated by the fact that, as for ground states, we can only hope for analytic solutions in very special cases. Therefore, in general, we must take recourse to  numerical approximation techniques in order to gain insight into the physics of a dissipative system. Typically the method of choice here is a \emph{Monte Carlo} sampling algorithm. Such methods have led to many insights into the dissipative systems occurring in quantum optics, but have faced limitations when applied to strongly interacting many particle systems, particularly fermions, due to the inevitable \emph{sign problem}.

There is, however, another general approach available to us, namely the \emph{variational method}. This method has been very successfully applied in the pure-state case leading to unparalleled insights into the equilibrium physics of strongly interacting many body systems. Further, the elegant \emph{time-dependent variational principle} (TDVP)~\cite{Dirac, TDVP} allows the \emph{locally optimal} study of nonequilibrium \emph{unitary} dynamics. The power of this method is well known in the field of quantum chemistry, where its application to the class of Hartree-Fock states is known as \emph{time-dependent Hartree Fock theory}~\cite{HFT_excitations}. This technique has also been exploited to great effect in the context of one-dimensional quantum spin systems in conjunction with powerful expressive variational classes such as matrix product states, a method synonymous with the \emph{density matrix renormalisation group} (DMRG) ~\cite{Jutho_dispersion}. 

In contrast to the pure-state case, there is no operationally unique way to formulate the variational method for mixed states because there is no distinguished measure of information distance: there are infinite families of inequivalent distance measures, including examples such as the \emph{fidelity} and the \emph{trace distance}. This has complicated the formulation of a mixed-state TDVP, which requires knowledge of the \emph{geometry} of state space. However, recent results in our understanding of the information geometry of mixed states allow us to revisit this problem. (See, however, ~\cite{Jackiw, Rajagopal} for related variational approaches to the von Neumann equation.)

Thus, in this Article, we formulate the TDVP for mixed states in the general case of distance measures arising from \emph{monotone riemannian metrics}. We then show that this method, when applied to the variational class of fermionic Gaussian states evolving according to an arbitrary Markovian CPT (completely positive trace-preserving) map $\rho_t = \mathcal{E}_t(\rho_0)$ are all equivalent to the application of Wick's theorem (via \emph{Gaussification}). Finally, we apply this method to the one-dimensional (1d) spinful Hubbard model subject to a decoherence process. 

\section{A review of the TDVP for pure states}
In this section we review of the TDVP for pure quantum states, and explain why this approach cannot be immediately applied to the mixed case. First, we present the necessary notation. We denote by $\mathcal{M}_n(\mathds{C})$ the set of all complex $n\times n$ matrices with entries in $\mathds{C}$. The state space of an $n$-dimensional quantum system is given by the set $\mathcal{D}_n$ (here viewed as a differentiable manifold) of all density operators defined by $\mathcal{D}_n =\{\rho \in \mathcal{M}_n(\mathds{C})\vert \rho^{\dagger} = \rho, \rho \geq 0, \tr(\rho)=1 \}$. In order to formulate a time-evolution within this manifold, we have to introduce the notion of tangent space. The tangent space $T_{\rho}\mathcal{D}_n$ to $\mathcal{D}_n$ at any interior point $\rho \in \mathcal{D}_n$ can be identified with the set $\{A \in \mathcal{M}_n(\mathds{C})\,\vert\, A^{\dagger} = A, \tr(A) = 0\}$ of traceless hermitian matrices (we assume that our processes never include any part of the boundary of the manifold). The set $\mathcal{D}_n$ can be given the structure of a Riemannian manifold by choosing a positive bilinear form $M_{\rho}(A,B)$ on $T_{\rho}\mathcal{D}_n$ for all $\rho \in \mathcal{D}_n$. This supplements us with the notion of a distance. Throughout we define a \emph{variational class}  to be a submanifold $\mathcal{V}$ of $\mathcal{D}_n$ parametrized according to $\mathcal{V} =  \{\rho(\mathbf{x})\,\vert\, \mathbf{x} \in \mathds{R}^D\}$, where the dependence on the parameters $x^j$ is assumed to be analytic.

We consider the time evolution of a quantum state $\rho_t = \rho(\mathbf{x}(t))\in \mathcal{D}_n$ in its most general form, i.e.\ $\rho_t \equiv \mathcal{E}_t(\rho_0)$, where $\mathcal{E}_t$ is a completely positive trace-preserving (CPT) map. Assuming that $\mathcal{E}_t$ is differentiable with respect to $t$ allows us to write the equation of motion as $\partial_t\rho_t = \mathcal{L}(\rho_t)$, where $\mathcal{L}$ is the infinitesimal generator of the dynamics. For example, $\mathcal{L}$ could describe a Hamiltonian or dissipative evolution of a quantum system. An exact solution to this evolution is in general hard to find, and so we aim at finding an optimal approximation to this evolution within the variational class $\mathcal{V}$. 

To this end, we first review the case of Hamiltonian time evolution within the set of pure quantum states. Here, the variational class of state vectors is represented as the set $\{\vert \psi(\mathbf{x})\rangle\vert \mathbf{x}\in \mathds{R}^D\}$. The time-dependent Schroedinger equation then reads $\dot x^{j}\partial_j |\psi(\mathbf{x})\rangle = -iH |\psi(\mathbf{x})\rangle$, where $\partial_j = \partial/\partial x^j$. Note that, in general, the vector $H |\psi(\mathbf{x})\rangle$ is not an element of the tangent space to $\mathcal{V}$ at $|\psi(\mathbf{x})\rangle$, whereas the left side is a linear combination of vectors that span the tangent space $T_{|\psi(\mathbf{x})\rangle}\mathcal{V}$. Thus, in general, there is no exact solution for $\dot x^j$. The best approximation is given by the solution to the minimisation of the \emph{information distance} (here measured using the \emph{fidelity}) between the left- and right-hand sides:  $\min_{\dot{x}^j}\| \dot x^{j}\partial_j |\psi(\mathbf{x})\rangle + iH |\psi(\mathbf{x})\rangle\|$. The minimum can be found by applying an orthogonal projection of $iH |\psi(\mathbf{x})\rangle$ onto the tangent space, as depicted in Fig.\ref{fig:TDVP}.

\begin{figure}[t]
\begin{center}
\includegraphics[width=0.5\columnwidth]{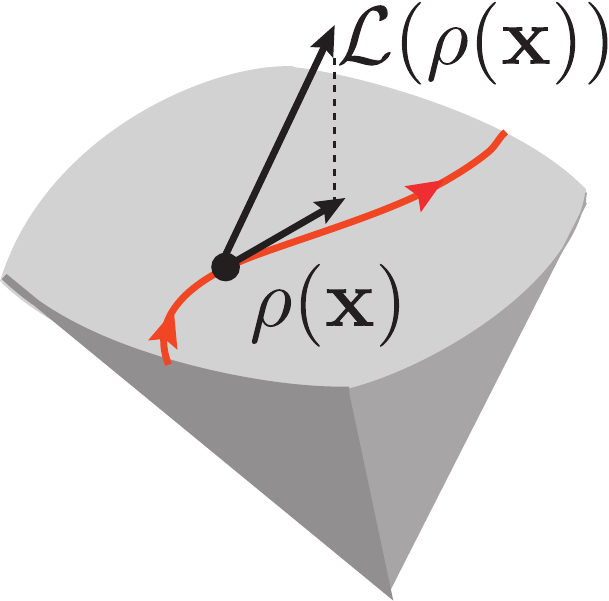}
\caption{Time-dependent variational principle with respect to a variational manifold of mixed states $\mathcal{V} = \{\rho(\mathbf{x})| \mathbf{x}\in \mathds{R}\}$. A mixed quantum state $\rho (\mathbf{x})$ within a variational manifold $\mathcal{V}$ evolves in time according to an arbitrary physical process described by a CPT map $\rho_t  = \mathcal{E}_t(\rho_0)$. In general, such a process leads out of the tangent space of the variational manifold at the point $\rho (\mathbf{x})$. Hence, we have to ``project back'' into $\mathcal{V}$ using an appropriate measure of distance (color online). \label{fig:TDVP} }
\end{center}
\end{figure}

This discussion immediately reveals why the TDVP cannot be directly applied to the mixed-state setting: the approximation of the RHS of $\dot x^j \partial_j \rho = \mathcal{L}_t(\rho)$ by a vector in the tangent space requires a unique notion of information distance. In the mixed-state case there is no operationally unique answer, since there exist infinite families of inequivalent measures, and hence, there is no canonical choice of Riemannian metric on $\mathcal{D}_n$. However, we explain in the next section a possible solution to this problem.

\section{Formulation of the TDVP for mixed states}
As we have explained above, there exists no canonical choice of Riemannian metrics in on the set of mixed quantum states. However, it turns out that there are several \emph{families} of Riemannian metrics which naturally arise from information-theoretic considerations. Here the natural condition is that the metric is \emph{monotone}, meaning that the norm induced by the bilinear form $M_{(\rho)}$ cannot increase under any CPT-map $\mathcal{E}$, i.e. $M_{\mathcal{E}(\rho)}(\mathcal{E}(A), \mathcal{E}(A))\leq M_\rho(A,A)$. The reasoning here is that the distinguishability of two states infinitesimally close to $\rho$ can never be increased under the action of a channel. Remarkably, Petz showed there is a one-to-one correspondence between the set of monotone metrics and a special class of \emph{superoperators} $\Omega_\rho$ (built in terms of convex operator functions) \cite{Petz}: these lead to monotone metrics according to $M_{\rho}(A,B) = \langle A, \Omega_{\rho}(B)\rangle \equiv \tr[A^{\dagger} \Omega_{\rho}(B)]$.

These monotone metrics now allow an operationally motivated formulation of the TDVP for the dissipative dynamics generated by $\partial_t \rho_t = \mathcal{L}(\rho_t)$ within a given variational class $\mathcal{V}$. The setup is identical to the pure-state case: we aim to find the optimal trajectory $\rho_t \in \mathcal{V}$ generated by the vector field coming from the optimal element $A \in T_{\rho_t}\mathcal{V}$ which is closest to $\mathcal{L}(\rho_t)$, where we use the quadratic form $M_{\rho_t}(A,B)$ to measure the distance. That is, we solve $\inf_{A\in T_{\rho_t}\mathcal{V}} M_{\rho_t}(A- \mathcal{L}(\rho_t),  A- \mathcal{L}(\rho_t))$.  An intuitive picture of this last equation is given in Fig.~\ref{fig:TDVP}: The evolution under $\mathcal{E}$ takes us out of the variational manifold $\mathcal{V}$, and we want to ''project back'' into $\mathcal{V}$ to find the state in the variational manifold that is the best approximation to this evolution. Using the definition of $M_{\rho_t}$ we find that this is equivalent to solving
\begin{align}
\inf_{A\in T_{\rho_t}\mathcal{V}} \langle A- \mathcal{L}(\rho_t), \Omega_{\rho_t}(A- \mathcal{L}(\rho_t))\rangle.
\end{align} 
Parametrizing $A = v^j \partial_j \rho_t$ we can rewrite this infimum as $\inf_{\mathbf{v}\in\mathds{R}^D} \mathbf{v}^T \mathbf{G}_{\rho}\mathbf{v} - \mathbf{v}^T\mathbf{l}_{\rho} - \mathbf{l}_{\rho}^T \mathbf{v} + c_0$. The solution is given by
\begin{align}\label{eq:TDVP}
\mathbf{v} &= \mathbf{G}^{-1}_{\rho}\mathbf{l}_{\rho},\\
(\mathbf{G}_{\rho})_{jk} &= \langle \partial_j \rho(\mathbf{x}(t)), \Omega_{\rho}(\partial_k \rho(\mathbf{x}(t)))\rangle,\\
(\mathbf{l}_{\rho})_j &= \langle \partial_j \rho(\mathbf{x}(t)), \Omega_{\rho}(\mathcal{L}(\rho(\mathbf{x}(t))))\rangle, 
\end{align}
where $\mathbf{G}_{\rho}$ is the \emph{pullback metric} or \emph{Gram matrix}. This solution gives us the optimal trajectory (locally in time) within $\mathcal{V}$ via integration of the equation of motion
\begin{equation}
	\partial_t\rho_t = v^j(t, \rho_t)\partial_j \rho_t.	
\end{equation}
This equation of motion is the first main contribution of our Article. Eqs. (2)--(5) can be applied to any variational manifold subject to any physical process. In the next Section we apply this framework  to a concrete example that is relevant in many-body physics.

\section{The TDVP for fermionic Gaussian States}
The understanding of fermionic quantum systems is of central interest many fields of physics. Fermions are building blocks of matter and thus central to some of the most fascinating effects known in the theory of many-body physics, like superconductivity or the quantum Hall effect. However, most problems of interest do not have a closed analytic solution, and we have to use numerical approximation techniques and we have to use appropriate variational wave functions to obtain insight into these systems. In the following we introduce the class of \emph{fermionic Gaussian states} (fGS) that has been successfully applied to solve fermionic many-body problems in the pure state setting. We show then how the TDVP can be applied using this class of states as our variational class $\mathcal{V}_G$ and give a numerical example in the last Section.

In the following we describe fermionic systems in terms of $N$ fermionic mode operators $a_j$ obeying the canonical anti-commutation relations $\{\ad_k, a_l\} = \delta_{kl}$. We use the equivalent representation in terms of $2N$ hermitian \emph{Majorana} operators $c_{2j-1} = \ad_j + a_j$ and $c_{2j} = (-i)(\ad_j - a_j)$ which obey $\{c_k, c_l\} = 2\delta_{kl}$. We take as our variational manifold the set of \emph{fermionic Gaussian states} (fGS). fGS are those states whose density operator can be expressed as an exponential of a quadratic function of the Majorana operators, $\rho = \kappa \exp[-\tfrac{i}{2}c^TKc]$, where $K = -K^T \in \mathds{R}^{2N \times 2N}$. All information about the state is encoded in the real and anti-symmetric covariance matrix (CM), $\Gamma_{kl} = \tfrac{i}{2}\tr([c_k, c_l]\rho)$, due to Wick's theorem: $i^p \tr[\rho c_{j_1}\ldots c_{j_{2p}}] = \mbox{Pf}(\Gamma_{j_1,
\ldots, j_{2p}})$ where $1 \leq j_1 < \ldots < j_{2p} \leq 2M$ and $\Gamma_{j_1,
\ldots, j_{2p}}$ is the corresponding $2p \times 2p$ submatrix of
$\Gamma$. $\mbox{Pf}(\Gamma_{j_1, \ldots, j_{2p}})^2= \mbox{det}
(\Gamma_{j_1, \ldots, j_{2p}})$ is called the Pfaffian (see, e.g., \cite{LinearOptics} for further details).

The class of fGS is a natural generalization of the variational classes used in Hartree-Fock and BCS-theory. Thus, it is combining and extending the most successful tools in the description of fermionic many-body systems and hence allows for a description of a wide class of fermionic phases of matter, like superfluids, Mott and spin ordered phases. Recently it has also been shown that fGS with topological order can be engineered in a cold-atom implementation via a local dissipative process~\cite{Diehl_Rico}. Thus, fGS have proven to be a powerful class capable of capturing fermionic phases with highly non-trivial properties.  Every pure fGS is the ground state of a quadratic Hamiltonian. Further, fGS remain Gaussian under the evolution according to a quadratic Hamiltonian or a dissipative process with linear Lindblad operators. Using this, fGS have allowed the approximation of the ground and thermal states of, the time-evolution~\cite{gHFTnum} of, as well as the excitation spectra~\cite{KrausOsborne} of interacting fermionic systems.  

The main ingredient used in all studies exploiting fGS is a process known as \emph{Gaussification}, i.e., the approximation of any $N$-body correlation function in terms of a product of single-particle correlation functions via Wick's theorem (see above):

\begin{definition}
 Let $\sigma$ be a fermionic quantum state. Then its \emph{Gaussification}, $\rho_G = \mathcal{G}(\sigma) \in \mathcal{V}_G$, is defined via the relation $\Gamma(\rho_G) = \Gamma(\sigma)$, which is equivalent to the application of Wick's theorem to $\sigma$. 
 \end{definition}
In the following, we show that the process of Gaussification is locally optimal in time within the variational class of fGS for all $\alpha$ metrics, since every monotone metric can be written as a convex combination of them~\cite{Petz}. 

\begin{theorem}
Let $\mathcal{E}_t$ be an arbitrary (differentiable) Markovian CPT map defining a time evolution on the space of density matrices via $\rho_t = \mathcal{E}_t(\rho_0)$. Then, the optimal approximation of this time evolution within the variational manifold of Gaussian states $\mathcal{V}_G$ with respect to any (convex combination of) monotone metrics of \emph{$\alpha$-norm type}:
\begin{align}\label{eq:omega}
\Omega_{\rho}^{\alpha}(\sigma) = \tfrac{1}{2}(\rho^{-\alpha}\sigma \rho^{\alpha-1} + \rho^{\alpha-1} \sigma \rho^{-\alpha}),
\end{align}
is obtained via Gaussification.
\end{theorem}

Theorem 1 is the second main result of this Article. To prove it we obtain a Gaussification of the time evolution of $\rho$ according to the generator $\mathcal{L}$ of the CPT map $\mathcal{E}_t$ as follows. Let $\rho(t+\delta t) = \rho(t) + \delta t \mathcal{L}(\rho)$, where $\delta t$ is an infinitesimal time step. The operator $\rho(t + \delta t) -\rho(t)$ is not necessarily a member of the tangent space $T_{\rho}\mathcal{V}_G$. However, the Gaussified operator $\mathcal{G}(\rho(t+\delta t)) -\rho(t)$ is, and it is therefore a linear combination of the tangent vectors $\partial_j \rho$, i.e.\  $\mathcal{G}(\rho(t+\delta t)) = \rho(t) + \delta t \sum_j v^j \partial_j \rho$, where $v^j \in \mathds{R}$. Since, by definition, Gaussification implies $\Gamma(\mathcal{G}(\rho)) = \Gamma(\rho)$, we obtain from the linearity of the map $\Gamma$, the following defining condition for the expansion parameters $v^j$:
\begin{align}\label{eq:Gaussification}
\tr[c_{k_1} c_{k_2} \mathcal{L}(\rho)] = \sum_jv^j \tr[c_{k_1} c_{k_2} \partial_j \rho].
\end{align} 
Now we show that an application of the TDVP projection with respect to any $\alpha$-norm also leads to Eq.~\eqref{eq:Gaussification}. For the proof of this statement we need the following 
\begin{corollary}
Let $\rho \in \mathcal{V}_G$ be a fGS of $N$ modes, and let $T_{\rho}\mathcal{V}_G$ denote the tangent space of $\rho$. Then, $\mathcal{B}=\{\partial_{(j_1,j_2)} \rho = (\Omega_{\rho}^{\alpha})^{-1}(ic_{j_1}c_{j_2})\}_{ 1\leq j_1< j_2 \leq 2N}$ is a hermitian basis for the tangent space  $T_{\rho}\mathcal{V}_G$. (Here $(j_1,j_2)$ denotes a \emph{multi index}.)
\end{corollary}
\begin{proof}
We first claim that if $\{Q_a\}_{a=1}^{2N}$ is a set of linearly independent operators then so is  $\{\Omega_{\rho}^{\alpha}(Q_a)\}_{a=1}^{2N}$. This follows immediately from the fact that $\Omega_\rho$ is a \emph{positive} superoperator so that $\mathcal{B}$ is a set of linearly independent vectors.

Next, we show that $\dim(\mathcal{B}) = \dim(T_{\rho}\mathcal{V}_G)$. To this end, we determine a basis of $T_{\rho}\mathcal{V}_G$ by applying the most general infinitesimal Gaussian transformation on $\rho$~\cite{LinearOptics}. These are of the from $\rho \mapsto W \rho W^{\dagger} / \tr[W \rho W^{\dagger}]$, where $W = e^{i \varepsilon Z_{kl}c_kc_l}$ with $Z_{kl} = X_{kl} + iY_{kl}$, $X_{kl}, Y_{kl} \in \mathds{R}$, and $\varepsilon \ll 1$. This leads directly to the tangent vectors $A^{(R)}_{kl} = [\rho, c_kc_l]$ and $A^{(I)}_{kl} = i\{\rho, c_kc_l\} - 2\Gamma_{kl}\rho$. In order to determine the number of linearly independent tangent vectors, we work in the basis $\tilde c_k =\sum_l O_{kl} c_l$, where $OO^T = \mathds{1}$, so that $\rho$ is in its standard form, $\rho = \prod_{j=1}^N\tfrac{1}{2}(\mathds{1} + i\lambda_j\tilde c_{2j-1}\tilde c_{2j})$, where $\lambda_j \in (-1, 1)$ (i.e. $\rho$ has no pure subspace). Then, the tangent vectors are readily calculated. For all $1 \leq k < l \leq N$ we obtain $A^{(R)}_{2k,2l} =$ $ -2i(\lambda_l \tilde c_{2k} \tilde c_{2l-1} + \lambda_k \tilde c_{2k-1} \tilde c_{2l}) \hat \rho_{kl}$, $A^{(R)}_{2k,2l-1} =$ $ 2i(\lambda_l \tilde c_{2k} \tilde c_{2l} - \lambda_k \tilde c_{2k-1} \tilde c_{2l-1})\hat \rho_{kl}$, $A^{(R)}_{2k-1,2l-1} =$ $ 2i(\lambda_l \tilde c_{2k-1} \tilde c_{2l} + \lambda_k \tilde c_{2k} \tilde c_{2l-1})\hat \rho_{kl}$ and $A^{(R)}_{2k-1,2l} =$ $ 2i(-\lambda_l \tilde c_{2k-1} \tilde c_{2l-1} + \lambda_k \tilde c_{2k} \tilde c_{2l})\hat \rho_{kl}$, where $\hat \rho_{kl} =$ $\prod_{j\neq k,l}\tfrac{1}{2}(\mathds{1} + i\lambda_j\tilde c_{2j-1}\tilde c_{2j})$. Further, we find that $A^{(I)}_{2k,2l} = $ $ 2i(\tilde c_{2k} \tilde c_{2l} - \lambda_k \lambda_l \tilde c_{2k-1} \tilde c_{2l-1})\hat \rho_{kl}$, $A^{(I)}_{2k,2l-1} =$ $ 2i(\tilde c_{2k} \tilde c_{2l-1} + \lambda_k \lambda_l \tilde c_{2k-1} \tilde c_{2l})\hat \rho_{kl}$, $A^{(I)}_{2k-1,2l-1} =$ $ 2i(\tilde c_{2k-1} \tilde c_{2l-1} - \lambda_k \lambda_l \tilde c_{2k} \tilde c_{2l})\hat \rho_{kl}$ and $A^{(I)}_{2k-1,2l} =$ $ 2i(\tilde c_{2k-1} \tilde c_{2l} + \lambda_k \lambda_l \tilde c_{2k} \tilde c_{2l-1})\hat \rho_{kl}$. Thus we obtain, for all $\lambda_{k,l} \in (-1,1)$, the four linearly independent basis vectors $i\tilde c_{2k-1}\tilde c_{2l-1}\hat \rho_{kl}$, $i\tilde c_{2k-1}\tilde c_{2l}\hat \rho_{kl}$, $i\tilde c_{2k}\tilde c_{2l-1}\hat \rho_{kl}$ and $i\tilde c_{2k-}\tilde c_{2l}\hat \rho_{kl}$. If $k=l$, we obtain $A^{(I)}_{2k-1,2k} =2 i(1-\lambda_k^2)\tilde c_{2k-1}\tilde c_{2k}\hat \rho_{k}$, while $A^{(R)}_{2k-1,2k}=0$. Thus the dimension of the tangent space is $\dim(T_{\rho}\mathcal{V}_G)=4 N(N-1)/2 + N = 2N(2N-1)= \dim(\mathcal{B})$.
 
Finally, we show that the operators $\partial_{(j_1, j_2)}\rho$ are hermitian. This follows immediately from the fact that $\Omega_{\rho}^{\alpha}((\partial_{(j_1, j_2)}\rho)^{\dagger}) = \Omega_{\rho}^{\alpha}(\partial_{(j_1, j_2)}\rho)$, and that $\Omega_{\rho}^{\alpha}$ is invertible
\end{proof}

With Corollary 1 in hand we now directly apply the TDVP projection in our special basis, following Eq. \eqref{eq:TDVP}. We see that
\begin{align*}
 (G_{\rho})_{(j_1,j_2),(k_1,k_2)} &= \langle \partial_{(j_1,j_2)}\rho, \Omega_{\rho}^{\alpha}(\partial_{(k_1,k_2)}\rho)\rangle \nonumber\\
 &= \tr [\partial_{(j_1,j_2)}\rho c_{k_1}c_{k_2}],\\
 (l_{\rho})_{(j_1,j_2)} &= \langle \partial_{(j_1,j_2)}\rho, \Omega_{\rho}^{\alpha}(\mathcal{L}(\rho))\rangle =\tr[\mathcal{L}(\rho)c_{k_1}c_{k_2}], 
 \end{align*}
 and we arrive at Eq. \eqref{eq:Gaussification}. This proves the equivalence of Gaussification and the application of the TDVP to the variational class of Gaussian states.
\begin{figure}[t]
\begin{center}
\includegraphics[width=0.95\columnwidth]{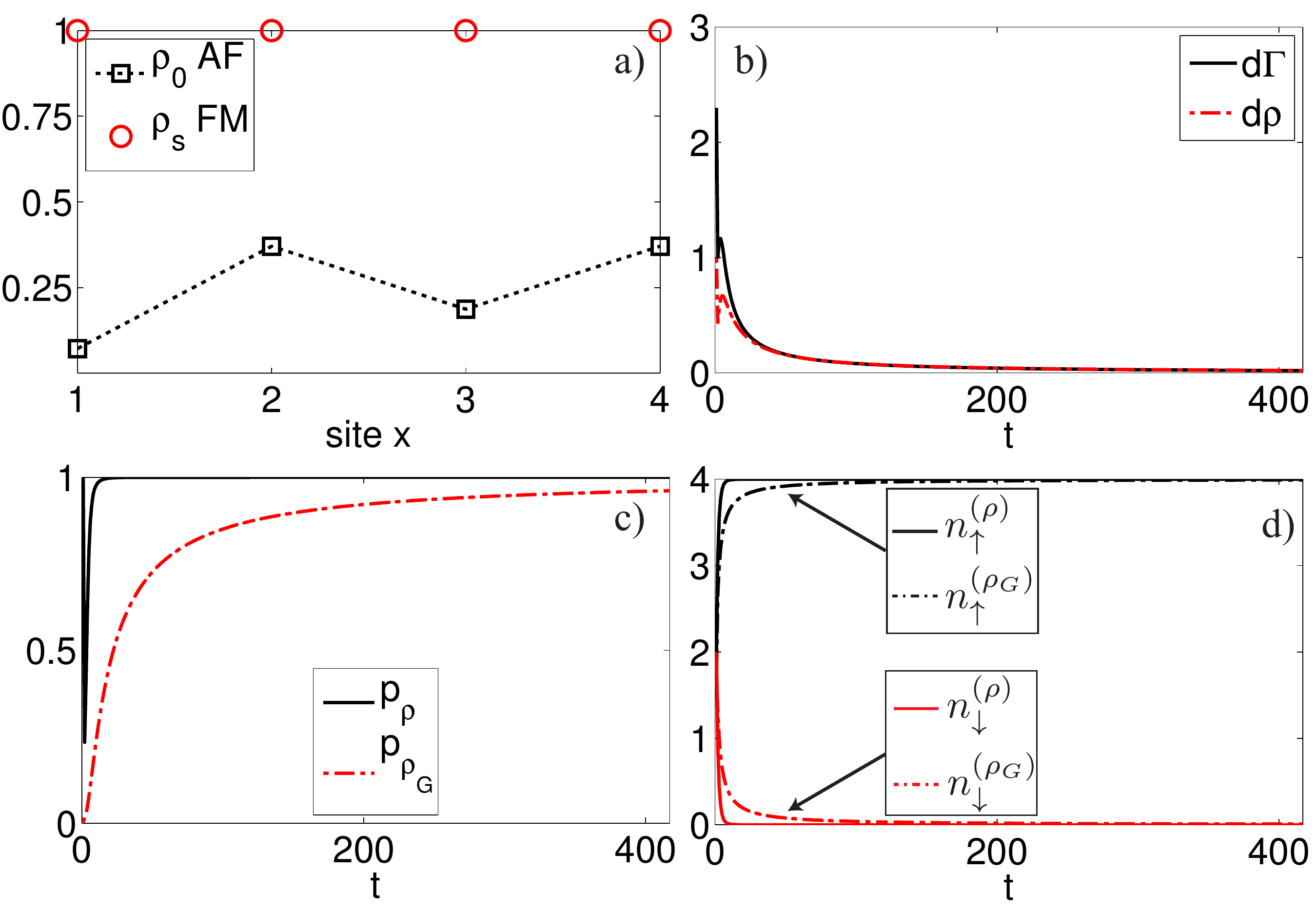}
\end{center}
\caption{a) Anti-ferromagnetic order (AF) in the ground state $\rho_0$ of the Hubbard model for $u=4, \mu = -2$. The steady state $\rho_s$ shows a ferromagnetic order (FM). b) Comparison of the real-time evolution (time $t$ in units of $1/\kappa$) of the dissipative process with the Gaussified version. We present the difference of the CM and the two states $d\Gamma(t) = ||\Gamma_{\rho_G(t)} - \Gamma_{\rho(t)}||_2$, $d\rho(t) = || \rho_G(t) - \rho(t)||_2$. c) Purity, $p_{\rho} = \tr[\rho^2]$,  for the exact (solid)  and the Gaussified (dashed) evolution, $\rho_G$. d) Evolution of the particle number $n$ for the two spin states $\sigma_s = \uparrow, \downarrow$ for the real ($\rho$) and Gaussified ($\rho_G$) process (color online).\label{fig:example}}
\end{figure}

\section{Example}
In the following we apply our approach to the $1d$ spinful Hubbard model with repulsive interactions subject to a magnetic field. The evolution of the system is described by a Lindblad equation $\partial_t \rho = -i[H,\rho] + \kappa\sum_{x}j_x \rho j_x^{\dagger} - \tfrac{1}{2}\{ j_x^{\dagger}  j_x, \rho\}$, where
\begin{align*} 
H &= J\sum_{x, \sigma_s} \ad_{x, \sigma_s}a_{x+1,\sigma_s} + u\sum_{x}n_{x,\uparrow}n_{x,\downarrow} + \mu \sum_{x,\sigma_s} n_{x,\sigma_s},\\
j_x &= \ad_{x,\uparrow}a_{x,\downarrow}. 
\end{align*}
Starting from the ground state of the Hubbard Hamiltonian we expect that the external noisy magnetic field modeled by the operators $j_{x}$ drives the system to a completely spin-polarized  state. The time scale of this process depends on the ratio between the decoherence strength $\kappa$ and the parameters of the Hubbard Hamiltonian. 

We consider a system of $L=4$ sites with periodic boundary conditions at half filling and consider an interaction $u = 4$, a chemical potential $\mu = -2$ and $\kappa = 1$, where we take the hopping $J$ to be the energy scale. The unique ground state $\rho_0$ has anti-ferromagnetic (AF) order (see Fig.\ref{fig:example}a  (squares)). We implement a real-time evolution $\rho(t)$ of the dissipative process, arriving at a unique steady state $\rho_s$ which is completely spin polarized (ferromagnetic order (FM)) (Fig. \ref{fig:example}a (circles)) and given by $\prod_{x=1}^4\ad_{x\uparrow}\vac $. In order to measure how well the Gaussified evolution given by $\rho_G(t)$ approximates the exact dynamics, we present  the deviation $d\Gamma(t) = ||\Gamma_{\rho_G(t)} - \Gamma_{\rho(t)}||_2$ of the CM of the real and Gaussified process and the  distance between the two states  $d\rho(t) = || \rho_G(t) - \rho(t)||_2$  in Fig.~\ref{fig:example}b.  We find that, as one might expect, on short time scales the Gaussified dynamics takes a different path from the exact evolution, but quickly coincides with the exact evolution for intermediate and long time scales. This can be explained by the fact that the ground state of the Hubbard model is not a Gaussian state, so that the initial state for the Gaussified evolution is a large distance from the ground state. 

In order to obtain more insight into the Gaussified evolution we compare the time dependence of some physical quantities with the exact evolution. In Fig.~\ref{fig:example} c we present the purity $p_{\rho} = \tr[\rho^2]$ for the exact evolution (solid) and the Gaussified process (dashed). We see that we end up in a pure state with only spin up particles (Fig.~\ref{fig:example}d) in the limit $t\rightarrow \infty$ in both cases. 

\section{Summary}
In summary, we have extended the time-dependent variational principle to the mixed-state setting, providing locally optimal (in time) equations that allow for an approximation of any Lindblad dynamics, given a variational manifold of mixed states and some information metric. In the case of fermionic Gaussian states we have proven that all $\alpha$-metrics lead to the same dynamics in the space of density matrices, which can equivalently be obtained via an application of Wick's theorem at each time step (Gaussification). Thus, this method can easily be applied to large systems in any dimension and geometry, providing a powerful numerical tool for a variational study of dissipative dynamics for mixed quantum states. 

\emph{Acknowledgments.---}%
We thank J. Haegeman and F. Verstraete for useful discussion. This work was supported, in part, by the cluster of excellence EXC 201 “Quantum Engineering and Space-Time Research”, by the Deutsche Forschungsgemeinschaft (DFG),  the EU grant QFTCMPS and by the Austrian Ministry of Science BMWF as part of the UniInfrastrukturprogramm of the Research Platform Scientific Computing at the University of Innsbruck.


\end{document}